\documentclass{article}

\usepackage[numbers, compress]{natbib}
\usepackage[final]{nips_2017}

\usepackage[utf8]{inputenc} 
\usepackage[T1]{fontenc}    
\usepackage{hyperref}       
\usepackage{url}            
\usepackage{booktabs}       
\usepackage{amsfonts}       
\usepackage{nicefrac}       
\usepackage{microtype}      

\usepackage{amsmath,amssymb,bbold}
\usepackage{bauer, bauer_e}
\usepackage{algorithm,algorithmic}
\usepackage{hyperref}

\newenvironment{proof}[1][Proof:\\]{\begin{trivlist}
\item[\hskip \labelsep {\bfseries #1}]}{\end{trivlist}$\Box$}

\usepackage{color}

\usepackage[pdftex]{graphicx}
\graphicspath{{figures/}}

\usepackage{bauer, bauer_e}

\title{Partial Optimality of  Dual Decomposition\\ for MAP Inference in Pairwise MRFs}

\author{
  Alexander Bauer$^{1,2}$, \hspace*{5pt}Shinichi Nakajima$^{1,2}$,\hspace*{5pt} Nico G{\"o}rnitz$^{2}$,\hspace*{5pt} Klaus-Robert~M{\"u}ller$^{1,2,3,4}$\\
  $^1$Berlin Big Data Center, Berlin, Germany\\
  $^2$Machine Learning Group, Technische Universit{\"a}t Berlin, Berlin, Germany\\
  $^3$Max Planck Institute for Informatics, Saarbr{\"u}cken, Germany\\
  $^4$Department of Brain and Cognitive Engineering, Korea University, Seoul, Korea\\
  \texttt{\{alexander.bauer, nakajima, nico.goernitz, klaus-robert.mueller\}@tu-berlin.de} \\
}

\begin{document}

\maketitle

\begin{abstract}
Markov random fields (MRFs) are a powerful tool
for modelling statistical dependencies for a set
of random variables using a graphical representation.
An important computational problem related to MRFs, called maximum a posteriori (MAP) inference,
is finding a joint variable assignment with the maximal probability.
It is well known that the two popular optimisation techniques for this task,
linear programming (LP) relaxation and dual decomposition (DD), have
a strong connection both providing an optimal solution to the MAP problem
 when a corresponding LP relaxation is tight.
However, less is known about their relationship in the opposite and more realistic case.
In this paper, we explain how the fully integral assignments obtained via DD
partially agree with the optimal fractional assignments via LP relaxation
when the latter is not tight.
In particular, for binary pairwise MRFs the corresponding result
suggests that both methods share the partial optimality property
of their solutions.
\end{abstract}

\section{Introduction}
The framework of graphical models such as Markov random fields (MRFs) \cite{Koller:2009:PGM:1795555, Pearl89, WainwrightJ08}
provides a powerful tool
for modelling statistical dependencies for a set
of random variables using a graphical representation.
It is of a fundamental importance for many practical application areas
including natural language processing, information retrieval, computational biology, and computer vision.
A related computational problem, called maximum a posteriori (MAP) inference, is 
finding a joint variable assignment with the maximal probability.
The vast amount of existing methods for solving the discrete MAP problem (see \cite{KappesAHSNBKKKL15} for an overview)
can be divided roughly into three different groups: methods based on graph-cuts, methods based on message passing,
and polyhedral methods.
The latter group is tightly connected to a popular approach of linear programming (LP) relaxation \cite{Sontag_thesis10, WainwrightJ08, opac-b1094316, opac-b1086766}, which is based
on reformulating the original combinatorial problem  
as an integer linear problem (ILP)
and then relaxing the integrality constraints on the variables. 
Besides providing a lower bound on the optimal value
its popularity is partially\footnote{It also provides an optimal solution for any tree-structured MRF.} due to the fact that
for binary pairwise MRFs with submodular energies it is guaranteed to find an optimal solution \cite{Koller:2009:PGM:1795555, WangY14}.
Unfortunately, for bigger problems with a high number of  variables and constraints it becomes
impractical due to the expensive consumption in memory and computation. 

As an alternative to the LP relaxation,
dual decomposition (DD) \cite{KomodakisPT11, SonGloJaa_optbook, RushC14, Everett_63, Luenberger73}
provides an effective parallelisation framework  for solving the MAP problem.
Furthermore, it has an appealing property that any found solution comes
with a certificate of optimality which allows for efficient evaluation
whether a corresponding variable assignment is primal optimal.
Finally, it is well known that the two optimisation techniques have a strong connection
both providing an optimal solution if the LP relaxation is tight.
However, less is know about their relationship in the opposite and more realistic case.
Instead the main focus of the existing literature is on tightening the standard LP relaxation \cite{KomodakisPT11, abs-1206-3288, Sontag_thesis10}.
In contrast, the aim of this paper is to investigate the connections between the two techniques
(in the original formulation)
if the LP relaxation is not tight.
More precisely, we focus on the following issue.
The main idea of DD is to decompose a given MRF
into different trees (or other subgraphs) on which inference can be performed efficiently
and trying to enforce an agreement on the overlapping variables between different trees to obtain global consistency.
If the LP relaxation is not tight, some trees will provide inconsistent assignments which disagree on the overlapping parts.
Here we analyse the nature of this disagreement and get
the following main results:
\begin{itemize}
\item given an optimal (fractional) solution of the LP relaxation, there always exists an optimal variable assignment via DD which agrees with the integral part of the LP solution
\item for binary pairwise MRFs in a non degenerate case\footnote{By a non degenerate case
we mean the case where the LP relaxation has a unique (fractional) solution. That is, the optimum is attained at a corner and not at an edge or a facet of
a corresponding polytope.}, the unambiguous part among all optimal assignments from different trees in a decomposition
coincides with the integral part of the (fractional) optimal assignment via LP relaxation.
\end{itemize}
We note that the first result holds also for non binary MRFs with arbitrary higher order potentials.
On the other hand, for binary pairwise MRFs it
implies that a corresponding assignment contains
the strongly persistent part\footnote{For binary pairwise MRFs, the integral part of an optimal solution of the LP relaxation is known to be strongly persistent.
That is, any optimal solution of a corresponding MAP problem must agree with this partial integral assignment.} of the LP solution \cite{WainwrightJ08, HammerHS84}.
The second result suggests a strategy how to extract the strongly persistent part
by looking at the intersection of all optimal assignments to the overlapping trees.
In this sense, both methods LP relaxation and DD share a partial optimality property of their solutions. 

Note that a corresponding fractional solution obtained via LP relaxation (for binary pairwise MRFs) is half integral.
That is, every fractional node is equal to $0.5$. This provides no information about the preferences
of the fractional variables in that case preventing use of rounding techniques.
In contrast, DD always provides a fully integral assignment
which inherits the strongly persistent part of the LP relaxation
making DD an appealing optimisation method.

Finally, we argue that depending on the final goal there is a decision to make with respect to the
degree of  a corresponding decomposition.
Usually, the main goal is to find an accurate (integral) assignment to the variables in a MRF. 
In that case a decomposition over spanning trees is more beneficial.
It significantly speeds  up the convergence but most importantly it is straightforward
how to extract an optimal assignment.
On the other hand if we want to extract the strongly persistent part, then a decomposition
over edges is more appropriate. In that case it is straightforward to extract the
unambiguous part, but at the same time it is NP-hard to construct an optimal (and globally consistent) assignment from the individual edges in a decomposition.


\section{Notation and Background}\label{sec:Notation}

\subsection{MAP Inference as an Optimisation Problem}
For a set of $n$ discrete variables $\bfx = \{x_1, ..., x_n\}$ taking values from a finite set $S$
we define the energy of a pairwise MRF factorising over a graph $G = (\mcV, \mcE)$ according to:
\begin{equation}
E(\bfx) = \sum_{i \in \mcV} \theta_i(x_i) + \sum_{(i,j) \in \mcE} \theta_{i,j}(x_i,x_j),
\end{equation}
where the functions $\theta_i(\cdot) \colon S \rightarrow \mathbb{R}$, $\theta_{i,j}(\cdot, \cdot) \colon S \times S \rightarrow \mathbb{R}$ denote the corresponding unary and pairwise potentials, respectively.
The \emph{maximum a posteriori} (MAP) problem, that is, computing an assignment with the highest probability
is equivalent to the problem of finding an assignment which minimises the energy.

Probably the most popular method for solving this problem is based on the linear programming (LP) relaxation technique.
For this purpose, the MAP problem is first represented as an (equivalent) integer linear problem (ILP):
\begin{equation}
\begin{aligned}
& \underset{\bfmu \in \mcX_G}{\text{minimise}}
& & \bftheta^\T \bfmu
\end{aligned}
\label{OP0}
\end{equation}
where $\bfmu$ corresponds to a joint variable assignment $\bfx$ in the standard overcomplete representation \cite{WainwrightJ08}.
That is, $\bftheta$ is a vector with entries $\theta_i(x_i)$ for all $i \in \mcV, x_i \in S$ and  $\theta_{i,j}(x_i,x_j)$ for all $(i,j) \in \mcE$, $x_i, x_j \in S$,
and $\bfmu$ is a binary vector of indicator functions for nodes and edges $\mu_i(x_i), \mu_{i,j}(x_i, x_j) \in \{0, 1\}$, where
$\mu_i(s) = 1 \Leftrightarrow x_i = s$ and $\mu_{i,j}(s_1,s_2) = 1 \Leftrightarrow x_i = s_1 \land x_j = s_2$.
The set $\mcX_G$ corresponds to all valid assignments of a pairwise MRF over a graph $G$ and has the following compact representation:
\begin{equation}
\label{E_X_G}
\mcX_G := \left\{ \bfmu \in \mathbb{R}^d \hspace*{5pt}
\begin{array}{|ll}
 \sum_{x_i} \mu_i(x_i) = 1 & \forall i \in \mcV\\
  \sum_{x_i} \mu_{i,j}(x_i, x_j) = \mu_j(x_j) & \forall (i,j) \in \mcE, \forall x_j \in S\\
  \sum_{x_j} \mu_{i,j}(x_i, x_j) = \mu_i(x_i) & \forall (i,j) \in \mcE, \forall x_i \in S\\
  \mu_i(x_i) \in \{0, 1\} & \forall i \in \mcV, \forall x_i \in S\\
  \mu_{i,j}(x_i, x_j) \in \{0, 1\} &  \forall (i,j) \in \mcE, \forall x_i, x_j \in S
\end{array}
\right\}
\end{equation}
A convex hull of this set, which we denote by $\mcM_G := \textrm{conv }\mcX_G$ plays a special role in the optimisation
and is known as the \emph{marginal polytope} of a corresponding MRF.
Namely, the problem (\ref{OP0}) is equivalent to the one where we replace
the set $\mcX_G$ by its convex hull $\mcM_G$, that is
\begin{equation}
\min_{\bfmu \in \mcX_G} \bftheta^\T \bfmu = \min_{\bfmu \in \mcM_G} \bftheta^\T \bfmu.
\end{equation}

Since finding an optimal solution of the above ILP
or equivalently minimising its linear objective over the marginal polytope
is in general intractable,
we usually consider the following relaxation:
\begin{equation}
\begin{aligned}
& \underset{\bfmu \in L_G}{\text{minimise}}
& & \bftheta^\T \bfmu
\end{aligned}
\label{OP1}
\end{equation}
where we optimise over a bigger set
$L_G \supseteq   \mcM_{G} \supseteq \mcX_{G}$ called the \emph{local consistency polytope} of a MRF over a graph $G$,
which results from relaxing the integrality constraints $\mu_i(x_i), \mu_{i,j}(x_i, x_j) \in \{0, 1\}$ in the definition of $\mcX_G$
by allowing the corresponding variables to take all real values in the interval $[0, 1]$.
That is,
\begin{equation}
L_G := \left\{ \bfmu \in \mathbb{R}^d \hspace*{5pt}
\begin{array}{|ll}
 \sum_{x_i} \mu_i(x_i) = 1 & \forall i \in \mcV\\
  \sum_{x_i} \mu_{i,j}(x_i, x_j) = \mu_j(x_j) & \forall (i,j) \in \mcE, \forall x_j\\
  \sum_{x_j} \mu_{i,j}(x_i, x_j) = \mu_i(x_i) & \forall (i,j) \in \mcE, \forall x_i\\
  \mu_{i,j}(x_i, x_j) \geqslant 0 &  \forall (i,j) \in \mcE, \forall x_i, x_j \in S
\end{array}
\right\}
\end{equation}
Note that the non-negativity of the unary variables $\mu(x_i) \geqslant 0$ implicitly
follows from the combination of the agreement constraints between node and edge variables and the non-negativity of the latter.

\subsection{Optimisation via Dual Decomposition}
We now briefly review the DD framework for MAP inference in (pairwise) MRFs \cite{KomodakisPT11}.
The main idea is to decompose the original intractable optimisation problem (OP) in (\ref{OP0}) over a graph $G$ into a set
of tractable inference problems over subtrees $\{\mcT_j\}_{j=1}^m$, $\mcT_j \subseteq G$,
which are coupled by a set of agreement constraints to ensure the consistency.
That is, each of the individual subproblems $j \in \{1, ..., m\}$ corresponds to the MAP inference on a subtree $\mcT_j = (\mcV_j, \mcE_j)$
of the original MRF.
More precisely, we define the following OP
\begin{equation}
\begin{aligned}
& \underset{\bfmu^1 \in \mcX_{\mcT_1}, ..., \bfmu^m \in \mcX_{\mcT_m}, \bfnu}{\text{minimise}}
& & \sum_{j=1}^m {\bftheta^j}^\T \bfmu^j\\
& \text{subject to}
& & \mu^j_i(x_i) = \nu_i(x_i) \hspace*{10pt} \forall j \in \{1, ..., m\}, \forall i \in \mcV_j, \forall x_i \in S\\
\end{aligned}
\label{OP2}
\end{equation}
where each vector $\bfmu^j$ denotes the variables of a local subproblem with respect to $\mcT_j$, and
$\bfnu$ is a set of global variables $\nu_i(x_i)$ on which the variables $\mu_i^j(x_i)$ of the (overlapping) subproblems must agree.
We can choose any decomposition with the only condition that the corresponding
trees together must cover all the nodes and edges of $G$, that is, $\mcV = \bigcup_{j=1}^m \mcV_j$ and $\mcE = \bigcup_{j=1}^m \mcE_j$,
as well as $\bftheta^\T \bfmu = \sum_{j=1}^m \bftheta_j^\T \bfmu^j$.
Note that OP in (\ref{OP2}) is equivalent to the ILP in (\ref{OP0}).
A corresponding LP relaxation given by
\begin{equation}
\begin{aligned}
& \underset{\bfmu^1 \in L_{\mcT_1}, ..., \bfmu^m \in L_{\mcT_m}, \bfnu}{\text{minimise}}
& & \sum_{j=1}^m {\bftheta^j}^\T \bfmu^j\\
& \text{subject to}
& & \mu^j_i(x_i) = \nu_i(x_i) \hspace*{10pt} \forall j \in \{1, ..., m\}, \forall i \in \mcV_j, \forall x_i \in S\\
\end{aligned}
\label{OP3}
\end{equation}
is equivalent to
the OP in (\ref{OP1})
in the sense that both have the same optimal value and the same optimal solution set.

In the corresponding dual problems the goal is
to maximise the
dual function of the OPs in (\ref{OP2}) and (\ref{OP3})
according to
\begin{equation}
\label{D2}
\begin{aligned}
& \underset{\bfu \in \mcU}{\text{maximise}}
& & g_{\ref{OP2}}(\bfu), \hspace*{2pt} \text{where} \hspace*{2pt} g_{\ref{OP2}}(\bfu) = \underset{\bfmu^1 \in \mcX_{\mcT_1}, ..., \bfmu^m \in \mcX_{\mcT_m}}{\inf} \left\{ \sum_{j=1}^m (\bftheta^j + \bfu^j)^\T \bfmu^j\right\}\\
\end{aligned}
\end{equation}
and
\begin{equation}
\label{D3}
\begin{aligned}
& \underset{\bfu \in \mcU}{\text{maximise}}
& & g_{\ref{OP3}}(\bfu), \hspace*{2pt} \text{where} \hspace*{2pt} g_{\ref{OP3}}(\bfu) = \underset{\bfmu^1 \in L_{\mcT_1}, ..., \bfmu^m \in L_{\mcT_m}}{\inf} \left\{ \sum_{j=1}^m (\bftheta^j + \bfu^j)^\T \bfmu^j\right\}\\
\end{aligned}
\end{equation}
respectively,
over a restricted set of dual values
\begin{equation}
\mcU := \left\{\bfu \colon \sum_{j \colon i \in \mcV_j} \bfu^j_i(x_i) = 0, i \in \{1, ..., n\}, x_i \in S \right\}
\end{equation}
We here overload the notation in the following sense.
The dual variables have the following form $\bfu = (\bfu^1, ..., \bfu^m)$, $\bfu^j = (..., u^j_i(x_i), ...)$.
Therefore, since we ignore the edges, the number of dual variables $\bfu^j$  is smaller than
the dimensionality of $\bfmu^j$ (or $\bftheta^j$).
However, for the algebraic operations (e.g.\ inner product) to make sense we implicitly assume that the vector $\bfu^j$ (if required) is appropriately filled with zeros to get the same dimensionality as $\bfmu^j$. 
A derivation of the above dual problems can be found in \cite{KomodakisPT11}.
There are different ways to solve a corresponding dual problem.
The most popular is a subgradient method for convex non differentiable objectives. Alternatively,
we could use a variant of block coordinate descent or cutting plane algorithm.

\section{Connections between LP Relaxation and DD}\label{sec:Main}
The facts summarised in Subsection \ref{s_1} are mainly known. We provide
them for the sake of completeness.
In Subsection \ref{s_2} we present new insights in the connections between the two optimisation techniques.
\subsection{Case 1: LP Relaxation yields an Integral Solution}\label{s_1}
It is well known that if the LP relaxation is tight, both the LP relaxation and DD provide an
integral optimal solution to the MAP problem.
From a different perspective, this means
that strong duality holds for the OP in (\ref{OP2}).
That is, there is zero duality gap between optimal values of OPs in (\ref{OP2}) and (\ref{D2}).
Equivalently, it implies the existence of consistent optimal assignments to subtrees according to a chosen decomposition.
We summarise these insights in the following lemma.
\begin{lemma}
\label{L_opt}
The following claims are equivalent:
\begin{flushleft}
(i) LP relaxation in (\ref{OP1}) has an integral solution\\
(ii) strong duality holds for problem (\ref{OP2})\\
(iii) $\bar{\mu}^{j_1}_i(x_i) = \bar{\mu}^{j_2}_i(x_i) \hspace*{10pt} \forall j_1, j_2 \in \{1, ..., m\}, i \in \mcV_{j_1} \cap \mcV_{j_2}, x_i \in S$
\end{flushleft}
where $\bar{\bfmu} = (\bar{\bfmu}^1, ..., \bar{\bfmu}^m)$ is a (not necessarily unique) minimiser of
the Lagrangian $\mcL(\cdot, ..., \cdot, \bfu^*)$ for OP in (\ref{OP2}) and $\bfu^*$ is a dual optimal.
\end{lemma}

\subsection{Case 2: LP Relaxation yields a Fractional Solution}\label{s_2}
If the LP relaxation is not tight, a corresponding optimal solution $\bfmu^*$ will have fractional components.
Given such a fractional solution, we denote by $\mcI \subseteq \{1, ..., n\}$ the indices of the variables $x_1, ..., x_n$, which have been assigned an integral value in $\bfmu^*$, and by $\mcF \subseteq \{1, ..., n\}$ the remaining set 
of indices corresponding to the fractional part.
Formally,  $i \in \mcI \Leftrightarrow \forall x_i \in S \colon \mu_i^*(x_i) \in \{0,1\}$
and $\mcF = \{1, ..., n\} \setminus \mcI$.
In contrast to LP relaxation, assignments produced via DD are fully integral.
Given a tree decomposition of a corresponding MRF, the optimal assignments to different subtrees, however,
will partially disagree on the overlapping parts.
Here we denote by $\mcA \subseteq \{1, ..., n\}$ the indices of the variables $x_1, ..., x_n$
with unique assignment. Formally, $i \in \mcA$, if for every tree which contains $x_i$ and every optimal assignment to that tree,
$x_i$ has the same value.
Similarly we denote by $\mcD \subseteq \{1, ..., n\}$ the set of indices for which at least
two different trees disagree on their optimal assignments, that is, $\mcD = \{1, ..., n\} \setminus \mcA$.
We now provide a formal analysis of the relationship between 
the sets $\mcI$ and $\mcA$, or equivalently between $\mcF$ and $\mcD$.

\begin{theorem}
Let $\bfmu^*$ be an optimal fractional solution of the LP relaxation (\ref{OP1}),
$\bar{\bfu}$ a dual optimal for OP in (\ref{OP2}), and $\mcL \colon \mcX_{\mcT_1} \times \mcX_{\mcT_m} \times \mcU \rightarrow \mathbb{R} $ a corresponding Lagrangian.
There always exists a set of minimisers $\bar{\bfmu}^1 \in \mcX_{\mcT_2}, ..., \bar{\bfmu}^m \in \mcX_{\mcT_m}$
of the Lagrangian $\mcL(\cdot, ...,  \cdot, \bar{\bfu})$
which agree with the integral part $\mcI$ of $\bfmu^*$, that is,
\begin{center}
$\forall j \in \{1, ..., m\}, i \in \mcI \cap \mcV_j, x_i \in S \colon \hspace*{10pt} \bar{\mu}^j_i(x_i) = \mu^*_i(x_i)$,
\end{center}
for short $\bar{\bfmu}^j_{\mcI} = \bfmu^*_{\mcI}$.
\label{T_1}
\end{theorem}
The result in the above theorem has the most intuitive interpretation in the case of a decomposition into spanning trees,
that is, when every tree $\mcT_j$ covers all the nodes ($\mcV_j = \mcV$) of the original graph.
In that case Theorem \ref{T_1} implies  that
for any optimal solution of the LP relaxation and a corresponding assignment $\bfx^*$ with an integral part $\mcI$, there exist
optimal assignments $\bfx^1, ..., \bfx^m$ (from a dual solution $\bar{\bfu}$) for the different spanning trees which agree on a set of nodes $\mcA$
with $x_i^j = x_i^*$ for all $i \in \mcI \subseteq \mcA$.
An immediate question arising is whether the two sets $\mcI$ and $\mcA$ are equal.
The answer is no.
In general, the two sets are not the same
and $\mcI$ will usually be a proper subset of $\mcA$.

Theorem \ref{T_1} motivates the following simple heuristic for getting an approximate integral solution,
which is especially suitable for a decomposition over spanning trees when using subgradient optimisation.
Namely, we can consider optimal assignments for every spanning tree and choose the best according
to the value of the primal objective. Increasing the number of trees in a decomposition also
increases the chance of finding a good assignment.
Furthermore, during the optimisation we can repeat this for the intermediate results after each iteration of
a corresponding optimisation algorithm
saving the currently best solution. Obviously, this only can improve the quality of the resulting assignment.
Note that it is the usual praxis with subgradient methods to
save the intermediate results since the objective is not guaranteed to improve in every step but can even get worse.
In that sense, the above heuristic does not impose additional computational cost.

Finally, it turns out that in a non degenerate case, where LP relaxation has a unique solution, the relationship $\mcI \subseteq \mcA$ (and therefore $\mcD \subseteq \mcF$) holds for all minimiser of a corresponding Lagrangian $\mcL(\cdot, ..., \cdot,\bar{\bfu})$
supported by the following theorem.

\begin{theorem}
\label{T_2}
Let $\bfmu^*$ be a unique optimal solution of the LP relaxation (\ref{OP1}),
$\bar{\bfu}$ a dual optimal for OP in  (\ref{OP2}), and $\mcL \colon \mcX_{\mcT_1} \times \mcX_{\mcT_m} \times \mcU \rightarrow \mathbb{R} $ a corresponding Lagrangian.
Each set of minimisers $\bar{\bfmu}^1 \in \mcX_{\mcT_2}, ..., \bar{\bfmu}^m \in \mcX_{\mcT_m}$ of the Lagrangian $\mcL(\cdot, ..., \cdot,\bar{\bfu})$
agrees with the integral part $\mcI$ of $\bfmu^*$, that is,
\begin{center}
$\forall j \in \{1, ..., m\}, i \in \mcI \cap \mcV_j, x_i \in S \colon \hspace*{10pt} \bar{\mu}^j_i(x_i) = \mu^*_i(x_i)$
\end{center}
for short $\bar{\bfmu}^j_{\mcI} = \bfmu^*_{\mcI}$.
\end{theorem}
In particular,
for binary pairwise MRFs this result 
implies that each assignment obtained via DD is partially  optimal
in a sense that it
always contains the strongly persistent part of the (fractional) solution of the LP relaxation.
Provided the fractional part is small,
it suggests
that the obtained assignments
will often have a low energy close to the optimum even if the LP relaxation is not tight.
This fact and the possibility of a parallel computation renders the dual decomposition
a practical tool for the MAP inference upon the LP relaxation.
We also note that the property $\mcI \subseteq \mcA$ in Theorem \ref{T_1} still holds for non binary MRFs with arbitrary higher order cliques,
but $\mcI$ is not guaranteed to be strongly persistent anymore.
In the following we build on an additional lemma.
\begin{lemma} 
Assume the setting of Theorem \ref{T_1}.
For every subproblem $j \in \{1, ..., m\}$ over a tree $\mcT_j$
there are minimisers $\bar{\bfmu}^j, \hat{\bfmu}^j \in \mcX_{\mcT_j}$,
where
\begin{equation}
\label{E_average}
\mu^*_i(x_i) = \frac{1}{2}(\bar{\mu}^j_i(x_i) + \hat{\mu}^j_i(x_i))
\end{equation}
holds for all $i \in \mcV_j, x_i \in S$.
\label{C_1}
\end{lemma}
The above lemma ensures that for each optimal solution $\bfmu^*$ of the LP relaxation,
for each tree in a given decomposition there always exist two different assignments
which agree exactly
on the nodes corresponding to the integral (strongly persistent) part $\mcI$ of $\bfmu^*$. 
It is also worth noting that when LP relaxation is not tight the sets of minimising assignments to the different trees 
are disjoint on the overlapping parts due to Lemma \ref{L_opt}.

Lemma \ref{C_1} and Theorem \ref{T_2} together imply that the unambiguous part among all optimal assignments from all trees in a decomposition coincides with the integral part of the (fractional) optimal assignment via LP relaxation
giving rise to the following theorem.
\begin{theorem}
\label{T_3}
Let $\bfmu^*$ be a unique optimal solution for the LP relaxation (\ref{OP1}) and
$\bar{\bfu}$ a dual optimal for OP in  (\ref{OP2}).
Then the unambiguous part $\mcA$ of optimal assignments among all the overlapping subproblems in a decomposition coincides with the integral part $\mcI$ of $\bfmu^*$.
That is,  $\mcA = \mcI$.
\end{theorem}
That is, we can extract the strongly
persistent part from DD by considering the intersection of optimal assignments for individual trees.
This is in particular convenient for a decomposition over single edges,
since computing the set of optimal assignments for an edge is straightforward.

\section{Related Work}
Several previous works \cite{WainwrightJW05a, Kolmogorov06, abs-1207-1395, RushC14}
including the original paper \cite{KomodakisPT11} on DD for MAP inference
comment on the optimality of DD and LP relaxation when the latter is tight.
While in \cite{KomodakisPT11} the authors make an explicit statement,
in \cite{WainwrightJW05a} the same question is approached by
introducing a notion of a tree agreement, which is generalised to a weak tree agreement in \cite{Kolmogorov06}.
Both, however, can be seen as a special case of \cite{KomodakisPT11} in the sense that the corresponding lower bound on the optimal value
of the MAP problem is equal to the optimal value of a corresponding dual problem in case of tree agreement
of all the involved subproblems.

Concerning the case where the LP relaxation is not tight,
to the best of our knowledge,
there are no previous works which directly address the question how the optimal solutions of the LP relaxation
are related to the assignments produced via DD.
Instead the existing works proceed with a discussion of tightening a corresponding relaxation \cite{KomodakisPT11, abs-1206-3288, Sontag_thesis10}.
In \cite{KolmogorovW05} the authors provide a related discussion for tree-reweighted (TRW) message passing and show that the partial assignment corresponding to the unambiguous part in the intersection of all optimal assignments for individual trees is strongly persistent. The TRW algorithms are less accurate than DD, however, in case of binary pairwise MRFs both achieve dual optimal value.
Our results are based on a different proof than in \cite{KolmogorovW05}, which supports a similar statement that the unambiguous part in a corresponding assignment is strongly persistent, but additionally implies that it is exactly the integral part of the  LP relaxation when the corresponding solution is unique.
Furthermore, our result in Theorem \ref{T_1} extends also to the case of arbitrary graphs.

Another  related question to the discussion in the present paper is the problem
of recovering optimal primal solutions of an LP relaxation from dual optimal solutions via DD \cite{KomodakisPT11, SonGloJaa_optbook, NedicO09, AnstreicherW09}.
In particular, \cite{AnstreicherW09} shows that the strongly persistent part of an LP solution (for a binary pairwise MRF)
can be recovered from the DD based on subradient optimisation.
However, it is not the case with other optimisation methods.
In contrast, our corresponding result holds independently of 
the chosen optimisation technique.

In the case of binary pairwise MRFs each optimal solution of the LP relaxation is 
known to be strongly persistent. In this paper we show that (in a non degrease case) this property is
inherited by every solution provided via DD in the sense that persistent part is a subset of the corresponding assignment.
Since the fractional part of a solution of the LP relaxation conveys no useful information
about the preference of the variables in the fractional area to be in specific state,
the presented result further supports the practical usefulness of the DD in that case.

\section{Numerical Validation}
 \begin{figure*}[t]
\centering
\includegraphics[scale = 0.8]{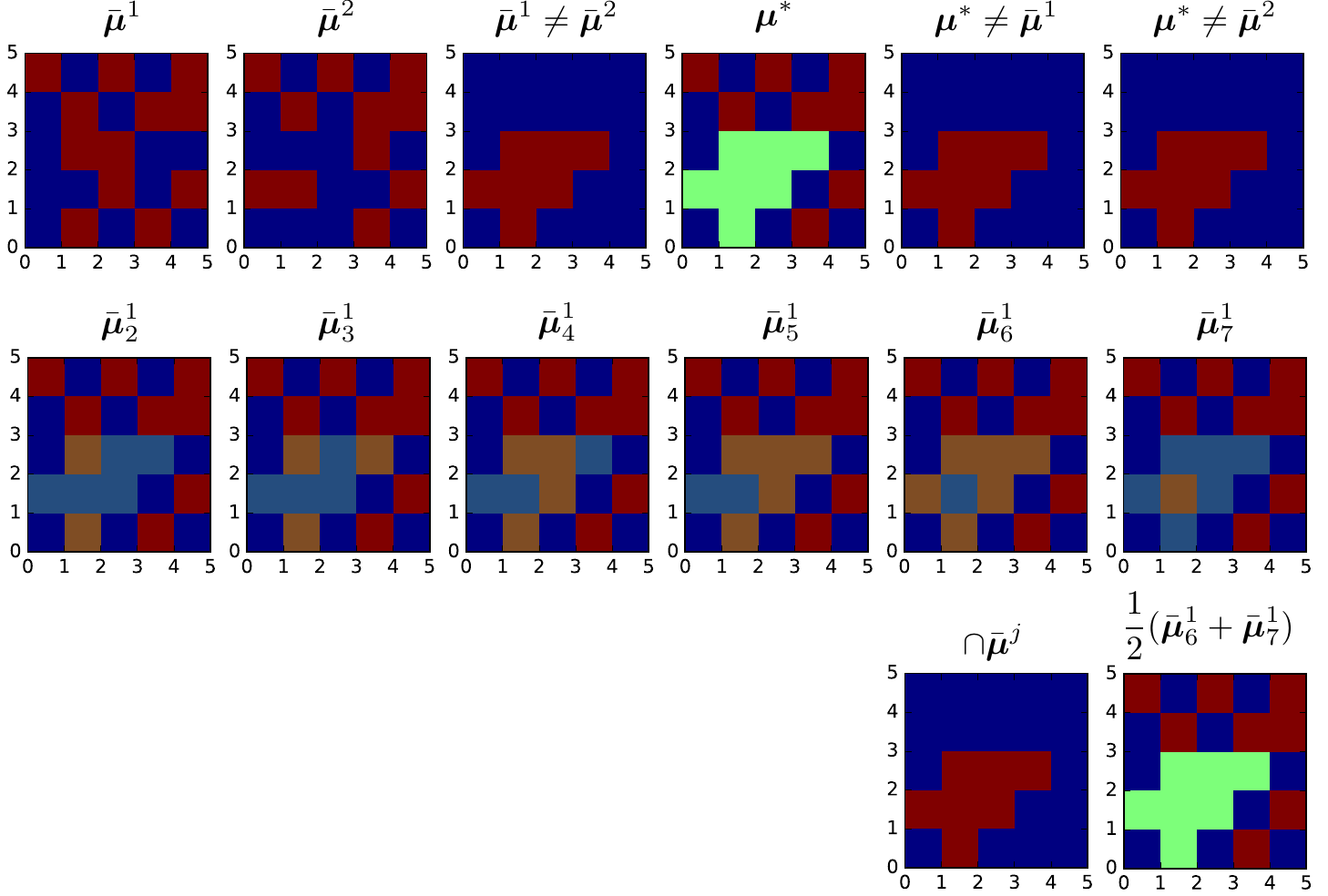}
\caption{Illustration of numerical validation of the theoretical results including Theorem \ref{T_1},  \ref{T_2},  \ref{T_3} and Lemma \ref{C_1}.
We consider a decomposition of an $5 \times 5$ Ising grid model
into two spanning trees (more precisely, disconnected forests): $\mcT_1$ consisting of all vertical edges and $\mcT_2$ consisting of all
horizontal edges. The corresponding optimal assignments for both trees are denoted by $\bar{\bfmu}^1$
and $\bar{\bfmu}^2$, respectively. ${\bfmu}^*$ is a corresponding optimal assignment from the LP relaxation.
Here, the red area corresponds to label $1$, blue area to label $0$, and green area to label $0.5$.
The individual plots in the first row visualise a corresponding assignment or a disagreement between two assignments.
The second row provides (additionally to $\bar{\bfmu}^1$) further optimal assignments for the subproblem $\mcT_1$.
The first plot in the third row illustrates the unambiguous part among all optimal and overlapping assignments in blue and the ambiguous part
in red.  The last plot illustrates an average of the two assignments above.
}
\label{fig_num_valid}
\end{figure*}
Here we present a numerical experiment summarised in Figure \ref{fig_num_valid} to validate the theoretical statements in the paper.
For this purpose we considered a $5 \times 5$ Ising grid model corresponding to 25 pixels
and defined its energy according to the following procedure.
The unary potentials have been all set to zero. The values of the corresponding edge potentials
have been selected uniformly at random from an interval $[-0.5, 0.5]$. This process has been repeated
until the corresponding LP relaxation yielded a fractional solution $\bfmu^*$ (see the fourth plot in the first row).

Given such an energy, we then considered a decomposition of the above Ising model into two spanning trees $\mcT_1$ (all the vertical edges)
and $\mcT_2$ (all the horizontal edges). We used the subgradient method to solve a corresponding dual problem
and got optimal assignments $\bar{\bfmu}^1$ and $\bar{\bfmu}^2$ corresponding to subproblems $\mcT_1$ and $\mcT_2$, respectively (see the first two plots in the first row).

The last two plots in the upper row in Figure \ref{fig_num_valid} can serve as a validation of Theorem \ref{T_1}.
Namely, the red area in these two plots corresponds to the variables on which each assignment disagrees with the LP solution $\bfmu^*$.
As we can see it happens only for the fractional area, where each variable in $\bfmu^*$ has the value $0.5$.
In other words, $\bar{\bfmu}^1$ and $\bar{\bfmu}^2$ agree with the integral part of $\bfmu^*$.

To support Theorem \ref{T_2} we computed further optimal assignments $\bar{\bfmu}^1_2, ..., \bar{\bfmu}^1_7$ for the subproblem $\mcT_1$ additionally to $\bar{\bfmu}^1$.
These are visualised in the second row in Figure \ref{fig_num_valid}. We overlay these assignments with the assignment $\bfmu^*$ in a transparent way to emphasise the difference to the LP solution.
We see that all these optimal assignments agree with the integral part of $\bfmu^*$.

The last plot in the third row validates the statement in Lemma \ref{C_1}.
It visualises the average of the the two plots above corresponding to the optimal assignments $\bar{\bfmu}^1_6$ and $\bar{\bfmu}^1_7$ for the subproblem $\mcT_1$.

Finally, the first plot in the third row supports
the claim in Theorem $\ref{T_3}$.
Namely, the blue area corresponds to the unambiguous part $\mcA$ where each variable has a unique value in all the optimal
assignments among the two subproblems. The red area marks the ambiguous part, where each variable has different
values in different assignments. We see that the equality $\mcA = \mcI$ holds.

\section{Conclusion}\label{sec:Conclusion}
We presented the established frameworks of linear programming (LP) relaxation and dual decomposition (DD)
for the task of MAP inference in discrete MRFs. In the case when a corresponding LP relaxation is tight, these two methods
are known to be equivalent both providing an optimal MAP assignment.
However, less is known about their relationship in the opposite and more realistic case.
While it is known that both methods have the same optimal objective value also in the non tight regime, it is an interesting question if there are other properties they share.
In particular, the connection between the solutions of LP relaxation and the assignments which can be extracted by DD has not been clarified.
For example, even if the solution of LP relaxation is unique, there might be multiple optimal (but disagreeing) assignments
via DD. What is  the nature of this ambiguity?  Are all these assignments equivalent or can we even extract additional information about the optimal
solutions from analysing the disagreement behaviour? These and other questions were the main motivation
for our paper.
Here we successfully provided a few novel findings explaining how the fully integral assignments obtained via DD
agree with the optimal fractional assignments via LP relaxation when the latter is not tight.
More specifically, we have proved:
\begin{itemize}
\item given an optimal (fractional) solution of the LP relaxation, there always exists an optimal variable assignment via DD which agrees with the integral part of the LP solution;
this also holds for non binary models with arbitrary higher order potentials
\item for binary pairwise MRFs (in a non degenerate case) the first result holds for every optimal assignment which can be extracted via DD
\item for binary pairwise MRFs (in a non degenerate case) the unambiguous part among all optimal assignments from different trees in a decomposition
coincides with the integral part of the (fractional) optimal assignment via LP relaxation
\end{itemize}
In particular, for binary pairwise MRFs the integral part of an optimal solution provided via
LP relaxation is known to be strongly persistent. Therefore, due to the properties listed above, 
we can conclude that (for this case) both methods LP relaxation and DD
share the partial optimality property of their solutions.

Practically, it has the following implications. 1) If the goal is to find an accurate MAP assignment 
we can use LP relaxation first to fix the integral part and then apply an approximation algorithm (e.g. loopy belief propagation) to set
the remaining variables in the fractional part -- each fractional node has the value $0.5$
showing no preference of a corresponding variable to be in a specific state. On the other hand,
DD provides fully integral assignments, which agree with the integral part of LP relaxation.
2) If we are only interested in the persistent part, we can extract the corresponding partial assignment
from DD by considering an intersection of all the optimal assignments to the individual trees in a decomposition.
The unambiguous part then coincides with the strongly persistent part of LP relaxation, provided a corresponding solution is unique,
or to a subset of it in the opposite case.
In particular, considering a decomposition over individual edges, is more beneficial in that case.
Namely, finding all optimal assignments to trees becomes a trivial task.

To summarise, the LP relaxation is a popular method for discrete MAP inference in pairwise graphical models
because of its appealing (partial) optimality properties.
However, this method does not scale nicely due to the extensive memory requirements
restricting its practical use for bigger problems.
Here, DD provides an effective alternative via distributed optimisation,
which scales to problems of arbitrary size --- we can always consider a decomposition over
the individual edges. Finally, the results presented in this paper suggest, that when using DD instead
of the LP relaxation we do not lose any of the nice properties of the latter. Both methods provide (for binary pairwise models) exactly the same 
information about their solutions.

An interesting question is which findings in the paper can be extended beyond the pairwise models
involving higher order potentials. We will continue this investigation in the future works.

\subsubsection*{Acknowledgments}
This work was supported by the Federal Ministry of Education and Research under the Berlin Big Data Center Project
under Grant FKZ 01IS14013A. The work of K.-R. M{\"u}ller was supported in part by the BK21 Program of NRF Korea,
BMBF, under Grant 01IS14013A and by Institute for Information and Communications Technology Promotion (IITP) grant funded by the Korea government (No. 2017-0-00451)

\small
\bibliographystyle{abbrv}
\bibliography{references}

\numberwithin{equation}{section}
\numberwithin{theorem}{section}
\numberwithin{figure}{section}
\numberwithin{algorithm}{section}
\numberwithin{table}{section}
\renewcommand{\thesection}{{\Alph{section}}}
\renewcommand{\thesubsection}{\Alph{section}.\arabic{subsection}}
\renewcommand{\thesubsubsection}{\Roman{section}.\arabic{subsection}.\arabic{subsubsection}}
\setcounter{secnumdepth}{-1}
\setcounter{secnumdepth}{1}
\setcounter{section}{0}

\newpage
\section*{Partial Optimality of  Dual Decomposition\\ for MAP Inference in Pairwise MRFs - Supplementary Material}

The following two lemmas are required for the proof of Theorem \ref{T_1}.
\begin{lemma}
\label{L:g1=g2}
Let $g_{\ref{OP2}}$, $g_{\ref{OP3}}$ be the dual functions of the problems (\ref{OP2}) and (\ref{OP3}) (presented in (\ref{D2}) and (\ref{D3})), respectively.
For any value of the dual variables $\bfu$ the equality $g_{\ref{OP2}}(\bfu) = g_{\ref{OP3}}(\bfu)$ holds.
\end{lemma}
\begin{proof}
\begin{equation*}
\begin{aligned}
g_{\ref{OP3}}(\bfu)   =& \inf_{\bfmu^1 \in L_{\mcT_1}, ..., \bfmu^m \in L_{\mcT_m}} \left\{\sum_{j = 1}^m (\bftheta^j + \bfu^j)^\T \bfmu^j  \right\}\\
  =& \inf_{\bfmu^1 \in L_{\mcT_1}} \{(\bftheta^1 + \bfu^1)^\T \bfmu^1 \} + ... + \inf_{\bfmu^m \in L_{\mcT_m}} \{(\bftheta^m + \bfu^m)^\T \bfmu^m\}\\
  =& \inf_{\bfmu^1 \in \mcX_{\mcT_1}} \{(\bftheta^1 + \bfu^1)^\T \bfmu^1 \} + ... + \inf_{\bfmu^m \in \mcX_{\mcT_m}} \{(\bftheta^m + \bfu^m)^\T \bfmu^m\}\\
  =& \inf_{\bfmu^1 \in \mcX_{\mcT_1}, ..., \bfmu^m \in \mcX_{\mcT_m}} \left\{\sum_{j = 1}^m (\bftheta^j + \bfu^j)^\T \bfmu^j  \right\} = g_{\ref{OP2}}(\bfu)
\end{aligned}
\end{equation*}
In the third equation we used the known fact that for a tree-structured MRF the local consistency polytope coincides
with the marginal polytope and that a corresponding objective for each subproblem is linear.
\end{proof}

\begin{lemma}
\label{L:Conv}
For any tree-structured MRF $\mcT$ with the corresponding set of valid assignments $\mcX_{\mcT} \subseteq \mathbb{R}^d$
in the standard overcomplete representation (as defined in (\ref{E_X_G}))
and for any $\mcI \subseteq \{1, ..., d\}$
the following equality holds
\begin{equation}
\textrm{conv }\mcX_{\mcT} \cap \{\bfmu \in \mathbb{R}^d \colon \bfmu_{\mcI} = \bfmu^*_\mcI\} = \textrm{conv }\{\bfmu \in \mcX_{\mcT} \colon \bfmu_{\mcI} = \bfmu^*_{\mcI}\},
\end{equation}
where $\bfmu^* \in \mathbb{R}^d$ is a point, for which there exists an assignment $\bar{\bfmu} \in \mcX_{\mcT}$ with $\bar{\bfmu}_{\mcI} = \bfmu^*_{\mcI}$.
\end{lemma}
\begin{proof}
$"\subseteq":$ Let $\bfv \in \textrm{conv } \mcX_{\mcT} \cap \{\bfmu \in \mathbb{R}^d \colon \bfmu_{\mcI} = \bfmu^*_{\mcI}\}$.
We now show that $\bfv$ can be represented as a convex combination of points $\bfmu \in \mcX_{\mcT}$ where $\bfmu_{\mcI} = \bfmu^*_{\mcI}$ for each $\bfmu$ in the combination.
On the one hand, since $\bfv \in \textrm{conv } \mcX_{\mcT}$ we
can write it as a convex combination $\bfv = \sum_{\bfmu} \alpha_{\bfmu} \bfmu$ for $\bfmu \in \mcX_{\mcT}$, $\alpha_{\bfmu} \in [0, 1]$
where we assume without loss of generality that the sum contains only $\alpha_{\bfmu} > 0$.
On the other hand, since $\bfv \in \{\bfmu \in \mathbb{R}^d \colon \bfmu_{\mcI} = \bfmu^*_{\mcI}\}$, it implies
that $\sum_{\bfmu} \alpha_{\bfmu} \bfmu_{\mcI} = \bfmu^*_{\mcI}$ must hold.
To prove the subset relationship
it suffices to show that $\mu_i = \mu^*_i$ for all $i \in \mcI$ for each of the points $\bfmu$
in the convex combination.
Now let $\mu^*_i = 1$ for some $i \in \mcI$.
Assuming that there is a $\hat{\bfmu}$
in our combination with $\hat{\mu}_i = 0$
results in the following contradiction:
\begin{equation*}
\sum_{\bfmu} \alpha_{\bfmu} \mu_i = \alpha_{\hat{\bfmu}} \cdot \hspace*{1pt}0 + \sum_{\bfmu \neq \hat{\bfmu}} \alpha_{\bfmu} \underbrace{\mu_i}_{\leqslant 1} \leqslant \sum_{\bfmu \neq \hat{\bfmu}} \alpha_{\bfmu} < 1 = \mu^*_i,
\end{equation*}
that is, $\sum_{\bfmu} \alpha_{\bfmu} \mu_i \neq \mu^*_i$.
Analogously, considering the case $\mu^*_i = 0$ and
assuming the existence of one $\hat{\mu}_i = 1$
gives rise to the following contradiction:
\begin{equation*}
\sum_{\bfmu} \alpha_{\bfmu} \mu_i = \alpha_{\hat{\bfmu}} \cdot 1 + \underbrace{\sum_{\bfmu \neq \hat{\bfmu}} \alpha_{\bfmu} \mu_i}_{\geqslant 0} \geqslant \alpha_{\hat{\bfmu}} >  0 = \mu^*_i,
\end{equation*}
that is, $\sum_{\bfmu} \alpha_{\bfmu} \mu_i \neq \mu^*_i$.\\[1pt]

\noindent $"\supseteq":$ This direction follows directly from $\{\bfmu \in \mcX_{\mcT} \colon \bfmu_{\mcI} = \bfmu^*_{\mcI}\} \subseteq \mcX_{\mcT}$ and
$\{\bfmu \in \mcX_{\mcT} \colon \bfmu_{\mcI} = \bfmu^*_{\mcI}\} \subseteq \{\bfmu \in \mathbb{R}^d \colon \bfmu_{\mcI} = \bfmu^*_{\mcI}\}$
where $\{\bfmu \in \mathbb{R}^d \colon \bfmu_{\mcI} = \bfmu^*_{\mcI}\}$ is a convex set.
\end{proof}

\section{Proof of Theorem \ref{T_1}}
The following derivations imply
the existence of a set of assignments $\bfmu^1, ..., \bfmu^m$ for the individual subproblems
according to the statement in the theorem:
\begin{equation*}
\begin{aligned}
\inf_{\bfmu^1 \in \mcX_{\mcT_1}, ..., \bfmu^m \in \mcX_{\mcT_m}} \mcL(\bfmu^1, ..., \bfmu^m, \bar{\bfu})  \overset{L \ref{L:g1=g2}}{=}& \inf_{\bfmu^1 \in L_{\mcT_1}, ..., \bfmu^m \in L_{\mcT_m}} \mcL(\bfmu^1, ..., \bfmu^m, \bar{\bfu})\\
  \overset{}{=}& \inf_{\bfmu^1 \in L_{\mcT_1}, ..., \bfmu^m \in L_{\mcT_m} \colon \bfmu^j_{\mcI} = \bfmu^*_{\mcI}} \mcL(\bfmu^1, ..., \bfmu^m, \bar{\bfu})\\
  \overset{L \ref{L:Conv}}{=}& \inf_{\bfmu^1 \in \mcX_{\mcT_1}, ..., \bfmu^m \in \mcX_{\mcT_m} \colon \bfmu^j_{\mcI} = \bfmu^*_{\mcI}} \mcL(\bfmu^1, ..., \bfmu^m, \bar{\bfu})\\
\end{aligned}
\end{equation*}
The first equality holds due to Lemma \ref{L:g1=g2}.
The second equality is due to the following fact.
Since strong duality holds for OP in (\ref{OP3}),
every optimal primal solution is a minimiser of the Lagrangian $\mcL(\cdot, ..., \cdot,\bar{\bfu})$.
Therefore, the set of feasible solutions restricted by the constraints $\bfmu^j_{\mcI} = \bfmu^*_{\mcI}$ contains at least one optimal
solution $\bfmu^1 := \bfmu^*|_{\mcT_1}, ..., \bfmu^m := \bfmu^*|_{\mcT_m}$,
where $\bfmu^*|_{\mcT_j}$ denotes a projection to a subspace corresponding to a tree $\mcT_j$.
The third equality can be shown using Lemma \ref{L:Conv} as follows:
\begin{equation*}
\begin{aligned}
   & \inf_{\bfmu^1 \in L_{\mcT_1}, ..., \bfmu^m \in L_{\mcT_m} \colon \bfmu^j_{\mcI} = \bfmu^*_{\mcI}} \mcL(\bfmu_1, ..., \bfmu_m, \bar{\bfu})\\
 =& \inf_{\bfmu^1 \in L_{\mcT_1}, ..., \bfmu^m \in L_{\mcT_m} \colon \bfmu^j_{\mcI} = \bfmu^*_{\mcI}} \left\{\sum_{j = 1}^m (\bftheta_j + \bar{\bfu}_j)^\T \bfmu^j  \right\}\\
 =& \inf_{\bfmu^1 \in L_{\mcT_1} \colon \bfmu^1_{\mcI} = \bfmu^*_{\mcI}} \{(\bftheta_1 + \bar{\bfu}_1)^\T \bfmu^1 \} + ... + \inf_{\bfmu^m \in L_{\mcT_m} \colon \bfmu^m_{\mcI} = \bfmu^*_{\mcI}} \{(\bftheta_m + \bar{\bfu}_m)^\T \bfmu^m\}\\
 \overset{(a)}{=}& \inf_{\bfmu^1 \in \mcM_{\mcT_1} \colon \bfmu^1_{\mcI} = \bfmu^*_{\mcI}} \{(\bftheta_1 + \bar{\bfu}_1)^\T \bfmu^1 \} + ... + \inf_{\bfmu^m \in \mcM_{\mcT_m} \colon \bfmu^m_{\mcI} = \bfmu^*_{\mcI}} \{(\bftheta_m + \bar{\bfu}_m)^\T \bfmu^m\}\\
 \overset{(b)}{=}& \inf_{\bfmu^1 \in \textrm{ conv }\{\bfmu \in \mcX_{\mcT_1} \colon \bfmu_{\mcI} = \bfmu^*_{\mcI}\}} \{(\bftheta_1 + \bar{\bfu}_1)^\T \bfmu^1 \} + ... + \inf_{\bfmu^m \in \textrm{ conv }\{\bfmu \in \mcX_{\mcT_m} \colon \bfmu_{\mcI} = \bfmu^*_{\mcI}\}} \{(\bftheta_m + \bar{\bfu}_m)^\T \bfmu^m\}\\
 \overset{(c)}{=}& \inf_{\bfmu^1 \in \mcX_{\mcT_1} \colon \bfmu^1_{\mcI} = \bfmu^*_{\mcI}} \{(\bftheta_1 + \bar{\bfu}_1)^\T \bfmu^1 \} + ... + \inf_{\bfmu^m \in \mcX_{\mcT_m} \colon \bfmu^m_{\mcI} = \bfmu^*_{\mcI}} \{(\bftheta_m + \bar{\bfu}_m)^\T \bfmu^m\}\\
=& \inf_{\bfmu^1 \in \mcX_{\mcT_1}, ..., \bfmu^m \in \mcX_{\mcT_m} \colon \bfmu^j_{\mcI} = \bfmu^*_{\mcI}} \left\{\sum_{j = 1}^m (\bftheta_j + \bar{\bfu}_j)^\T \bfmu^j  \right\}\\
 =& \inf_{\bfmu^1 \in \mcX_{\mcT_1}, ..., \bfmu^m \in \mcX_{\mcT_m} \colon \bfmu^j_{\mcI} = \bfmu^*_{\mcI}} \mcL(\bfmu_1, ..., \bfmu_m, \bar{\bfu})\\
\end{aligned}
\end{equation*}
where the step in $(a)$ holds because for every tree-structured MRF $\mcT_j$ the marginal polytope $\mcM_{\mcT_j}$ coincides
with the local consistency polytope $L_{\mcT_j}$; in step $(b)$ we use $\mcM_{\mcT_j} = \textrm{conv } \mcX_{\mcT_j}$ and Lemma \ref{L:Conv};
finally, the step in (c) holds because a linear objective over a polytope always achieves its optimum at least at one of the extreme points
(that is, corners) of the latter.

$\Box$

Note that the agreement on the integral part holds also for edge marginals,
that is, for every dimensions in $\bfmu^*$ with an integral value, even if the nodes of an edge are fractional.
Namely, we ca extend the
constraint $\bar{\bfmu}^j_{\mcI} = \bfmu^*_{\mcI}$ in Theorem \ref{T_1} to every dimension having an integral value
and the proof still works.

\section{Proof of Lemma \ref{C_1}}
We denote by $\bfmu^*|_{\mcT_j}$ a projection of a solution $\bfmu^*$ over a graph $\mcG$
to a subspace corresponding to a subtree $\mcT_j$. 
Since strong duality holds for the problem (\ref{OP3})
any primal optimal is a minimiser of the Lagrangian.
Therefore, each restriction $\bfmu^*|_{\mcT_j}$ is a minimiser of a corresponding subproblem
over tree $\mcT_j$.
Furthermore, Theorem \ref{T_1} guarantees an existence of a minimiser $\bar{\bfmu}^j$ that
agrees with the integral part of $\bfmu^*|_{\mcT_j}$ and differs from $\bfmu^*$ only on the fractional entries.
Note that this holds also for edge marginals. Namely, in the proof of Theorem \ref{T_1} we can extend the
constraints $\bar{\bfmu}^j_{\mcI} = \bfmu^*_{\mcI}$ to every dimension in $\bfmu^*$ having an integral value
(including edge marginals)
and the proof still works.
Any point on the line through these two solutions ($\bfmu^*|_{\mcT_j}$ and $\bar{\bfmu}^j$) is also optimal since the corresponding objective is linear.
We now show an existence of a corresponding solution $\hat{\bfmu}^j$ by construction.
We define
\begin{equation}
\label{E_C1}
\hat{\mu}^j_{i}(x_i) := \begin{cases}
  \mu^*_{i}(x_i),  & \text{if }i \in \mcI\\
  1-\bar{\mu}_i^j(x_i), & \text{otherwise.}
\end{cases}
\end{equation}
for each $i \in \mcV_j$
and
\begin{equation}
\label{E_C2}
\hat{\mu}^j_{i,k}(x_i, x_k) := \begin{cases}
  \mu^*_{i,k}(x_i, x_k),  & \text{if }\mu^*_{i,k}(x_i, x_k) \in \{0, 1\}\\
  1-\bar{\mu}_{i,k}^j(x_i, x_k), & \text{if } \mu^*_{i,k}(x_i, x_k) = 0.5
\end{cases}
\end{equation}
for each $(i,k) \in \mcE_j$.
It is easy to see that the above definition of $\hat{\bfmu}^j$ satisfies the equation (\ref{E_average}).
Furthermore, $\hat{\bfmu}^j$ lies on the line through $\bar{\bfmu}^j$ and the restriction $\bfmu^*|_{\mcT_j}$
and is therefore optimal.
We now show that it is feasible, that is, $\hat{\bfmu}^j \in \mcX_{\mcT_j}$.
Note that the variables $\hat{\bfmu}^j$ are either equal to the variables in $\bfmu^*_{\mcI}$ or
have an opposite value to the variables in $\bar{\bfmu}^j$. Therefore, they inherit 
the integrality constraints as well as the normalisation constraints from $\bfmu^*$ and $\bar{\bfmu}^j$.
Similar argument can be used for the marginalisation constraints.
More precisely, for the integral edges, where both end nodes are integral,
the marginalisation constraints hold true. We only need to check the cases with non integral edges.
This can be done by considering all the cases listed in Lemma \ref{L_frac_sol}.
We show exemplary one case -- the remaining cases are straightforward.
In the following we drop the superscript $j$ denoting the subproblem and write only $\hat{\bfmu}$ and $\bar{\bfmu}$.
Consider the case (a) from Lemma \ref{L_frac_sol}. First, we have
$\mu_{i,j}^*(0,0) = \mu_{i,j}^*(1,1) = 0.5$ and $\mu_{i,j}^*(0,1) = \mu_{i,j}^*(1,0) = 0$.
That is, $\bar{\mu}_{i,j}(0,1) = \bar{\mu}_{i,j}(1,0) = 0$. Without loss of generality assume
$\bar{\mu}_{i,j}(0,0) = 1$ and $\bar{\mu}_{i,j}(1,1) = 0$, that is, $\bar{\mu}_{i}(0) = \bar{\mu}_{j}(0) = 1$
and $\bar{\mu}_{i}(1) = \bar{\mu}_{j}(1) = 0$. Due to the construction in (\ref{E_C1})
we get $\hat{\mu}_{i}(0) = \hat{\mu}_{j}(0) = 0$ and $\hat{\mu}_{i}(1) = \hat{\mu}_{j}(1) = 1$.
This corresponds to $\hat{\mu}_{i,j}(0,0) = \hat{\mu}_{i,j}(0,1) = \hat{\mu}_{i,j}(1,0) = 0$ and $\hat{\mu}_{i,j}(1,1) = 1$
which is exactly what we get from construction in (\ref{E_C2}).
Therefore, we get a valid labelling of an edge and the marginalisation constraints are satisfied.
Finally, note that the equality
$\bfmu^*|_{\mcT_j} = \frac{1}{2}(\bar{\bfmu}^j + \hat{\bfmu}^j)$
also holds (including the edge marginals).

$\Box$

\section{Proof of Theorem \ref{T_2}}
Let $\bfmu^*$ be a unique optimal solution of the LP relaxation (\ref{OP1}).
We  use the notation $\bfmu^*|_{\mcT_j}$ to denote a reduction of $\bfmu^*$ to a corresponding subtree.
Theorem \ref{T_1} guarantees an existence of minimisers $\bar{\bfmu}^1 \in \mcX_{\mcT_1}, ..., \bar{\bfmu}^m \in \mcX_{\mcT_m}$
of a corresponding Lagrangian $\mcL(\cdot, ...,  \cdot, \bar{\bfu})$ where each $\bar{\bfmu}^j$ agrees with the integral part of $\bfmu^*|_{\mcT_j}$.
We now assume that there is another minimiser $\hat{\bfmu}^j$ for the $j$-th subproblem  with $\bar{\bfmu}^j_i(x_i) \neq \hat{\bfmu}^j_i(x_i)$
for some $i \in \mcI$, $x_i \in S$ and show that this assumption leads to a contradiction.
We do this by constructing another optimal solution of the LP relaxation different from $\bfmu^*$.

Assume for simplicity a decomposition over individual edges.
We consider the following relabelling procedure starting with an edge $(i,k)$
corresponding to the $j$-th subproblem above.
Since $\hat{\bfmu}^j$ and $\bar{\bfmu}^j$ both are minimisers for the corresponding subproblem,
the average $\tilde{\bfmu}^j := \frac{1}{2}(\hat{\bfmu}^j + \bar{\bfmu}^j)$ is also a minimiser (because the objective is linear)
and $\tilde{\bfmu}^j_i(x_i) = 0.5$ for $x_i \in S$.
That is, the $i$-th node is now assigned with a fractional label $0.5$.
The remaining nodes $x_r$ ($r \neq k$) adjacent  to $x_i$ can be relabelled in a consistent way to $x_i = 0.5$ such that a corresponding assignment $(0.5, x_r)$ is optimal for the edge $(i, r)$
by using the weak tree agreement property\footnote{Optimal assignments obtained via DD are known to satisfy the weak tree agreement (WTA) condition \cite{KomodakisPT11}.
In particular, for our purposes we use the following fact. Consider any two trees $\mcT_i$ and $\mcT_j$ which share a node $x_k$. Then for
any optimal configuration $\bfmu^i$ there exists an optimal configuration $\bfmu^j$ with $\mu^i_k(x_k) = \mu^j_k(x_k)$.}.
Namely, since there are two optimal assignments for the edge $(i,k)$ with both values for $i$,
for every adjacent edge $(i,k)$ there must also be optimal assignments with both values for $i$.
Therefore, we can define a new labelling for each edge adjacent to $i$ by computing the average of the corresponding assignments.
During this procedure some nodes $x_k$ can change their label.
Note that this is possible only for nodes with integral value in $\bfmu^*$.
To validate this claim 
consider a fractional node $x_k$. Since $x_i$ is integral in $\bfmu^*$,
there must be (due to lemma \ref{C_1}) optimal assignments $(x_i^*,0)$ and $(x_i^*,1)$ for edge $(i,k)$,
where $x_i^*$ is the optimal label of $x_i$ according to $\bfmu^*$.
Furthermore, because of $x_i = 0.5$ (due to relabelling $\tilde{\bfmu}^j$)
there also must be an optimal assignment for that edge of the form $(1-x^*,0)$ or $(1-x^*,1)$. In any case we can find an optimal average such that
$x_i = x_k = 0.5$. That is, the value of fractional $x_k$ does not change!

If a node $x_k$ changes his label to $0.5$ during this procedure, we then need to consider all its neighbours (except $x_i$) and proceed with the relabelling process.
More precisely, we have the following cases:\\
We now consider an edge $(i,k)$ where $x_i$ has been relabelled to $0.5$ in previous steps.\\
Case 1: $I \rightarrow I$\\
That is,  $x_i$ and $x_k$ both have an integral value in $\bfmu^*$.\\
$(a)$ $x_k$ does not change by computing a corresponding average, then there is nothing more to do.\\
$(b)$ $x_k$ changes. We label it with $0.5$ and
consider all adjacent cases (except $x_i$).\\
Case 2: $I \rightarrow F$\\
That is, $x_i$ is integral in $\bfmu^*$ and $x_k$ is fractional.
There must be (due to lemma \ref{C_1})
optimal assignments $(x_i^*, 0)$ and $(x_i^*, 1)$.
Because $x_i = 0.5$ now there must be (due to WTA) an optimal assignment $(1-x_i^*,0)$ or $(1-x_i^*,1)$
such that a corresponding optimal average results in $x_i = x_k = 0.5$.
So the label of $x_k$ does not change.\\
Case 3: $I \rightarrow I/F$\\
That is, $x_k$ has an integral value in $\bfmu^*$ but has been relabelled to $0.5$ previously.
Due to the WTA there are always assignments such that a corresponding average results in $x_i = x_k = 0.5$.

\noindent Since only integral nodes can change their label during the above relabelling procedure, there are no other cases to consider.
The relabelling procedure terminates with a new consistent jont labelling $\tilde{\bfmu}$ different from $\bfmu^*$.
We can prove the statement for arbitrary tree decompositions (not only over edges) by using similar arguments.

$\Box$

\section{On the fractional solutions of LP relaxation}
For binary pairwise MRFs the LP relaxation has the property that  in every (extreme) optimal solution
each fractional node is half integral \cite{Deza97, Sontag_thesis10}.
Furthermore, each edge marginal is either integral or has fractional values.
More precisely, an edge marginal is integral only if both end nodes are integral.
In fact, there are six further cases for fractional edge marginals
as specified in the following lemma.
\begin{lemma}
\label{L_frac_sol}
Let $\bfmu \in L_{G}$ be an extreme point. Then each edge marginal $\mu_{i,j}(x_i, x_j)$ is either integral (if 
both end nodes $x_i$ and $x_j$ are integral) or\\ 
(a) is equal to
\begin{center}
$\begin{array}{|c|c|c|}
  \hline
  \mu_{i,j}(x_i, x_j) & x_j = 0 & x_j = 1\\
  \hline
  x_i = 0 & 0.5 & 0\\
  \hline
  x_i = 1 & 0 & 0.5\\
  \hline
\end{array}$
\hspace*{10pt}or\hspace*{10pt}
$\begin{array}{|c|c|c|}
  \hline
  \mu_{i,j}(x_i, x_j) & x_j = 0 & x_j = 1\\
  \hline
  x_i = 0 & 0 & 0.5\\
  \hline
  x_i = 1 & 0.5 & 0\\
  \hline
\end{array}$
\end{center}
if both $x_i$ and $x_j$ are fractional;\\
(b) is equal to
\begin{center}
$\begin{array}{|c|c|c|}
  \hline
  \mu_{i,j}(x_i, x_j) & x_j = 0 & x_j = 1\\
  \hline
  x_i = 0 & 0.5 & 0\\
  \hline
  x_i = 1 & 0.5 & 0\\
  \hline
\end{array}$
\hspace*{10pt}or\hspace*{10pt}
$\begin{array}{|c|c|c|}
  \hline
  \mu_{i,j}(x_i, x_j) & x_j = 0 & x_j = 1\\
  \hline
  x_i = 0 & 0 & 0.5\\
  \hline
  x_i = 1 & 0 & 0.5\\
  \hline
\end{array}$
\end{center}
if $x_i$ is fractional and $x_j$ is integral ($x_j = 0$ on the left and $x_j = 1$ on the right);\\
(c) is equal to
\begin{center}
$\begin{array}{|c|c|c|}
  \hline
  \mu_{i,j}(x_i, x_j) & x_j = 0 & x_j = 1\\
  \hline
  x_i = 0 & 0.5 & 0.5\\
  \hline
  x_i = 1 & 0 & 0\\
  \hline
\end{array}$
\hspace*{10pt}or\hspace*{10pt}
$\begin{array}{|c|c|c|}
  \hline
  \mu_{i,j}(x_i, x_j) & x_j = 0 & x_j = 1\\
  \hline
  x_i = 0 & 0 & 0\\
  \hline
  x_i = 1 & 0.5 & 0.5\\
  \hline
\end{array}$
\end{center}
if $x_j$ is fractional and $x_i$ is integral ($x_i = 0$ on the left and $x_i = 1$ on the right);
\label{L_frac_sol}
\end{lemma}
\begin{proof}
The integral case is clear.
We now assume that a given edge is non integral, that is, at least one of the nodes is fractional.
First we show that in every case a matrix corresponding to an edge assignment contains only two different values $a$ and $b$.\\
Case (a):\\
Since every feasible solution $\bfmu \in L_{G}$ is subject to the marginalisation constraints $ \sum_{x_i} \mu_{i,j}(x_i, x_j) = \mu_j(x_j)$ 
and $\sum_{x_j} \mu_{i,j}(x_i, x_j) = \mu_i(x_i)$
the following equations must hold
\begin{align}
\label{4E}
\begin{split}
 \mu_{i,j}(0,0) + \mu_{i,j}(0,1) &= \mu_i(0)\\
 \mu_{i,j}(1,0) + \mu_{i,j}(1,1) &= \mu_i(1)\\
 \mu_{i,j}(0,0) + \mu_{i,j}(1,0) &= \mu_j(0)\\
 \mu_{i,j}(0,1) + \mu_{i,j}(1,1) &= \mu_j(1) 
\end{split}
\end{align}
Due to $\mu_i(0) = \mu_i(1) = \mu_j(0) = \mu_j(1) = 0.5$ it follows from (\ref{4E})
that $a:= \mu_{i,j}(0,0) = \mu_{i,j}(1,1)$ and $b:= \mu_{i,j}(0,1) = \mu_{i,j}(1,0)$.
Now we argue that $a, b \in \{0, 0.5\}$.
For this purpose assume that the edge marginal $\mu_{i,j}(x_i,x_j)$
contains other than half-integral values. So w.l.o.g. let $a \in (0, 0.5)$,
then also  $b \in (0, 0.5)$ (otherwise $a+b \neq 0.5$).
We now define two different feasible solutions $\bfmu^1$ and $\bfmu^2$ which have the same entries
as $\bfmu$ except the entries for the marginal $\mu_{i,j}(x_i,x_j)$,
which we define for $\bfmu^1$ by $a_1 := a + \epsilon$,  $b_1 := b - \epsilon$
and for $\bfmu^2$ by $a_2 := a - \epsilon$ and $b_2 := b + \epsilon$,
where $\epsilon$ is small enough such that $a_1, a_2, b_1, b_2 \in (0, 0.5)$.
Furthermore, due to $a_1+b_1 = a_2+b_2 = 0.5$ a corresponding edge assignment is feasible,
and therefore the solutions $\bfmu^1$, $\bfmu^2$.
Since $\bfmu = \frac{1}{2}(\bfmu^1 + \bfmu^2)$, the solution $\bfmu$ is 
a convex combination of two different feasible solutions, and is therefore not extreme contradicting our assumption
that $\bfmu$ is a corner of the local polytope. So it must hold $a, b \in \{0, 0.5\}$.
Finally, $a = 0$ implies $b = 0.5$ and vice versa due to $a+b=0.5$.
The remaining cases in $(b)$ and $(c)$ can be dealt with by using similar arguments as above.
\end{proof}

\end{document}